\documentclass[journal]{IEEEtran}
\IEEEoverridecommandlockouts
\usepackage{float}
\usepackage{setspace}
\usepackage{graphicx}
\usepackage{cite}
\usepackage{lipsum}
\usepackage{cuted}
\usepackage{amsmath,amssymb,amsfonts}\allowdisplaybreaks
\usepackage{mathrsfs}
\usepackage{amsthm}
\usepackage{algorithm}
\usepackage{algorithmic}
\graphicspath{{./figures/}}
\usepackage{textcomp}
\usepackage{verbatim}
\usepackage{dsfont}
\usepackage{soul}
\usepackage{enumerate}
\usepackage{diagbox}
\usepackage{bm}
\usepackage{booktabs} 
\usepackage{tabularx} 
\usepackage{ragged2e} 
\usepackage{xcolor}

\newcommand{\cA}{\mathcal{A}}\newcommand{\cC}{\mathcal{C}}\newcommand{\cD}{\mathcal{D}}

\newcommand{\cO}{\mathcal{O}}\newcommand{\cP}{\mathcal{P}}
\newcommand{\cS}{\mathcal{S}}

\newcommand{\bxi}{\boldsymbol \xi}

\newtheorem{myth}{Theorem}

\newtheorem{myle}[myth]{\bf Lemma}

\begin{document}

\title{Sustainable Placement with Cost Minimization in Wireless Digital Twin Networks}
\author{
{Yuzhi Zhou,~\IEEEmembership{Student Member,~IEEE}, Yaru Fu,~\IEEEmembership{Member,~IEEE},~Zheng Shi,~\IEEEmembership{Member,~IEEE}\\
~ Kevin Hung,~\IEEEmembership{Senior Member,~IEEE},~Tony Q. S. Quek,~\IEEEmembership{Fellow,~IEEE},~Yan Zhang,~\IEEEmembership{Fellow,~IEEE}}
\thanks{
This work was supported in part by the Hong Kong Research Matching Grant (RMG) in the Central Pot under Project No. CP/2022/2.1, in part by the Research and Development Fund (R\&D Fund) under reference No. RD/2023/1.8, in part by the Team-based Research Fund under Project No. TBRF/2024/1.10, in part by the National Natural Science Foundation of China under Grant 62171200, in part by Guangdong Basic and Applied Basic Research Foundation under Grant 2023A1515010900, in part by the National Research Foundation, Singapore and Infocomm Media Development Authority under its Future Communications Research $\&$ Development Programme, and in part by the EU Horizon 2020 Research and Innovation Programme under the Marie Sklodowska-Curie grant agreement No. 101008297. This article reflects only the authors' view. The European Union Commission is not responsible for any use that may be made of the information it contains. \emph{(Corresponding author: Yaru Fu)}

Y. Zhou, Y. Fu, and K. Hung are with the School of Science and Technology,  Hong Kong Metropolitan University, Hong Kong, 999077, China (e-mail: s1315400@live.hkmu.edu.hk; yfu@hkmu.edu.hk; khung@hkmu.edu.hk).

Z. Shi is with the School of Intelligent Systems Science and Engineering, Jinan University, Zhuhai 519070, China (e-mail: zhengshi@jnu.edu.cn).

Tony Q. S. Quek is with the Singapore University of Technology and Design, Singapore 487372 (e-mail: tonyquek@sutd.edu.sg).

Y. Zhang is with the Department of Informatics, University of Oslo (e-mail: yanzhang@ieee.org).
}
}

\markboth{}%
{Shell \MakeLowercase{\textit{et al.}}: Bare Demo of IEEEtran.cls for IEEE Journals}
\maketitle

\begin{abstract}
Digital twin (DT) technology has a high potential to satisfy different requirements of the ever-expanding new applications. Nonetheless, the DT placement in wireless digital twin networks (WDTNs) poses a significant challenge due to the conflict between unpredictable workloads and the limited capacity of edge servers. In other words, each edge server has a risk of overload when handling an excessive number of tasks or services. Overload risks can have detrimental effects on a network's sustainability, yet this aspect is often overlooked in the literature.
In this paper, we aim to study the sustainability-aware DT placement problem for WDTNs from a cost minimization perspective.
To this end, we formulate the DT placement-driven cost optimization problem as a chance-constrained integer programming problem.
For tractability, we transform the original non-deterministic problem into a deterministic integer linear programming (ILP) problem using the sample average approximation (SAA) approach. We prove that the transformed problem remains NP-hard and thus finding a global optimal solution is very difficult. To strike a balance between time efficiency and performance guarantee, we propose an improved local search algorithm for this ILP by identifying high-quality starting states from historical search data and enhancing the search process. Numerical results show a lower cost and higher efficiency of our proposed method compared with the previous schemes.
\end{abstract}


\section{Introduction}
\label{sec1}

The development of the sixth-generation (6G) mobile network has a high potential for the emerging Internet-of-Everything (IoE) applications \cite{DT_Magzine_ZangYan}. These applications, such as extended reality (XR), intelligent transportation, and haptic technologies, play a crucial role in advancing modern society. However, the successful implementation of these applications requires meeting specific critical requirements. These requirements may include ubiquitous and robust instant connectivity, extremely low latency, and enhanced edge intelligence capability of efficiently managing vast amounts of data \cite{DT_Architecture1}. To address these tremendous challenges, it is of high necessity to develop a self-sustaining paradigm and a proactive online-learning-based wireless system \cite{DT_Magzine}. In this context, the wireless digital twin network (WDTN), an innovative architecture based on the digital twin (DT) technique, provides a promising solution to meet these challenges \cite{DT_Architecture1}. Therein, a DT represents the digital replica of a physical entity, which is typically created at edge or cloud servers using historical data and real-time operational information. Through continuous interaction and synchronization with the physical counterpart, DT enables close monitoring and optimization of the physical system \cite{DT_Architecture1}. By harnessing these properties, WDTN has the potential to bridge the gap between physical systems and digital spaces, aligning with the goals envisioned by the 6G network.

It is not beyond our expectation that WDTN has attracted significant attention from researchers as a promising approach to address various challenges \cite{DT_Migration3,DT_Zhangyan5,DT_Zhangyan4,DT_Zhangyan2,DT_application2,DT_MEC1,DT_MEC2,DT_MEC3}.
In \cite{DT_Migration3}, X. Chen \emph{et al.} proposed DT-enabled mobile edge computing (MEC) networks to solve the service migration problem, where DTs were used to predict future traffic demands.
A distributed traffic prediction approach was proposed, and its performance was evaluated in a real-world mobile dataset. Based on this approach, a cooperative coefficient migration algorithm was used to reduce the migration cost and enhance the quality of service (QoS). In \cite{DT_Zhangyan5}, Y. Dai \emph{et al.} integrated DTs into the vehicular network to adaptively manage offloading policies and minimize offloading latency. The DTs were created by capturing and recording the essential features and real-time status of physical entities. To address the offloading problem, deep reinforcement learning (DRL) was employed, utilizing the DT observations as inputs for decision-making.
Likewise, in \cite{DT_Zhangyan4}, W. Sun \emph{et al.} proposed a lightweight DT-empowered air-ground network architecture to tackle the challenges of learning efficiency in aerial networks.
Building upon this framework, a learning efficiency maximization problem was formulated. Afterward, an incentive mechanism was proposed to solve this problem. Additionally, a scenario was presented in \cite{DT_Zhangyan2} by them where drones worked as aggregators in a federated learning (FL) network. Similar to the approach described in \cite{DT_Zhangyan4}, DTs were utilized to capture the essential features of ground-based physical devices and drones. Both static and dynamic incentives were employed to adjust the round of global updates to enhance the learning efficiency in these scenarios.
In \cite{DT_application2}, L. Lei \emph{et al.} used DTs to establish an unmanned aerial vehicle (UAV) swarm intelligent network. A case study was presented, demonstrating how a UAV swarm utilized DTs for synchronization and cooperation in tracking multiple vessel targets within marine scenarios.
Furthermore, in \cite{DT_MEC1}, W. Sun \emph{et al.} employed DTs to describe the features of physical devices, facilitating the FL processes of industrial WDTN. The frequency of FL aggregation was optimized to minimize the loss value of the FL procedure. To solve this optimization problem, DRL techniques were utilized.
In \cite{DT_MEC2} and \cite{DT_MEC3}, Y. Lu \emph{et al.} proposed architectures of DT-driven edge networks to strengthen communication efficiency. The combination of blockchain and FL was used for DT implementation.
Specifically, in \cite{DT_MEC2}, they adjusted the offloading policies of physical devices to minimize the loss function while enhancing the efficiency of the WDTN. They adopted DRL to solve the optimization problem with network features as observations. Whilst, in \cite{DT_MEC3}, they optimized the subcarrier allocation of the model transmission, as well as the participating of global aggregation and its corresponding time slot, to minimize the computation latency. A heuristic algorithm was proposed to solve the subcarrier allocation problem.


Several studies focus on privacy concerns in WDTNs.
In \cite{DT_Zhangyan3}, L. Jiang \emph{et al.} defined an optimization problem to investigate the trade-off between time consumption and learning accuracy in the FL process within a WDTN. A directed acyclic graph (DAG) blockchain-based method was adopted to establish the model update chain, enhancing security and privacy preservation during DT construction.
Similarly, in \cite{DT_place2}, Y. Lu \emph{et al.} addressed a DT offloading problem where DTs performed model training at the edge using FL. To enhance system security and data privacy, they utilized a permission blockchain. A multi-agent DRL algorithm was proposed to find a solution that balances learning accuracy and time cost.
On the other hand, the issue of DT association, which includes DT migration and DT placement, is a critical challenge that significantly impacts the performance of WDTNs. Several approaches have been explored to tackle the migration aspect of the association problem, as discussed in
\cite{DT_Migration1,6G_DigitalTwin,DT_Zhangyan1}.
In \cite{DT_Migration1}, Q. Liu \emph{et al.} studied a migration problem where DTs were utilized to characterize network elements and network slicing. A distributed approximating policy with the prediction of resource requirements was adopted to minimize the weighted ratio of energy consumption and network loading in this problem.
In \cite{6G_DigitalTwin}, Y. Lu \emph{et al.} proposed a DT migration problem in DT-enabled MEC networks. They employed DRL to optimize DT migration, aiming to minimize the overall system latency while satisfying the latency constraints of physical devices.
Furthermore, in \cite{DT_Zhangyan1}, W. Sun \emph{et al.} proposed an offloading problem in the WDTN, where service migration occurs when physical devices alter their target edge servers. The objective was to minimize the offloading delay associated with the migration by adjusting the target edge server of the physical device. A DRL-based offloading algorithm was employed to solve this problem.

While the literature discussed above demonstrates the effectiveness of WDTNs, it is essential to acknowledge their practical limitations. The sustainability of DT systems is crucial, especially for systems operated by edge servers. Specifically, it plays a pivotal role in defining operational excellence and long-term viability, as this criterion reflects the quality of continuous interaction between the physical and digital spaces.
In response to this critical need, it's important to design a new DT placement scheme. Through the deployment of such a scheme, organizations can ensure that the digital plane receives top-tier infrastructure support, consequently enhancing the overall sustainability of the system. To achieve this, several researchers have focused on the DT placement problem to improve system sustainability.
In \cite{DT_place1}, M. Vaezi \emph{et al.} addressed the DT placement problem to minimize the maximum delay in data request responses, thereby improving system sustainability. They formulated this problem as an integer quadratic program (IQP) and proposed a polynomial time approximation algorithm.
Furthermore, in \cite{DT_place3}, D. Wang \emph{et al.} proposed a joint optimization problem that included DT placement and resource allocation in sustainable computing networks. They aimed to minimize system delay and energy consumption while ensuring system reliability. To tackle this problem, they employed a DRL-based algorithm.
Another DT placement problem was formulated in \cite{DT_place4}, which focused on guaranteeing data freshness in the WDTN to enhance system sustainability. This problem was solved by maximizing a utility metric constructed using the concept of age of information (AoI). J. Li \emph{et al.} utilized an approximation algorithm to address this problem.
Besides, J. Li \emph{et al.} also formulated a DT placement problem for minimizing system delay in \cite{DT_place5} to improve system sustainability by incorporating serverless technology. An online algorithm was proposed to guarantee the system's performance.

Although the above research investigates the effect of placement on the sustainability control of WDTNs, the consideration of system overload in this context is commonly lacking.
Specifically, placing a large number of DT components on the same edge server can lead to CPU overload, particularly when the load of each component cannot be accurately estimated.
This uncertainty adds complexity to the placement decision-making problem. To be more specific, the risks of overload are challenging to eliminate as the network continues to operate.
Ultimately, this risk can result in placement failures within the system, thereby hindering the network's sustainability.
The most typical situation is that too many components are handled by the same edge server, resulting in an overloaded CPU, which further slows down the computing for each component or eventually fails the computation.
Moreover, the cost of edge server downtime can be painful for industries when edge servers are relied on to run their businesses \cite{whyNotComputation1,whyNotComputation2}.
It is worth noting that the stochastic formulation in the DT system, particularly for sustainability consideration, differs from the stochastic part of conventional MEC networks \cite{riskAware3,SAA_MEC_placement1,SAA_MEC_placement2}, where the latter typically focus on QoS enhancement or consumer behavior estimation.
In particular,
in \cite{riskAware3}, H. Badri \emph{et al.} studied an application placement problem with the consideration of the uncertainty of future QoS in the system. The objective function, which was maximized by current and expected QoS, was optimized by employing the sample average approximation (SAA) method and a graph-based placement algorithm.
Furthermore, in \cite{SAA_MEC_placement1}, Z. Ning \emph{et al.} proposed a service placement problem by considering the uncertain nature of user mobility. The authors adopted the Lyapunov and SAA methods to decouple the long-term optimization problem and relax the constraints of service execution time, followed by proposing a dynamic service placement scheme.
Besides, in \cite{SAA_MEC_placement2}, another service placement problem was formulated by H. Zhao \emph{et al.}, where the stochastic property of mobile devices on the service composition scheme was considered. The SAA method was utilized to make the problem trackable, and the authors employed a genetic algorithm to minimize the response time during the service placement.

To address challenges brought by the sustainability placement in DT systems, it becomes crucial to propose a method that incorporates sustainability considerations when operating a WDTN, ensuring that the failure risk of the network, while possibly inevitable, remains controllable.
In other words, this risk should be kept below an acceptable threshold. These observations serve as a motivation for our work. 
In this paper, we aim to consider sustainability while operating WDTNs. Our goal is to minimize the overall cost of DT placement while taking into account the constraint of restricting the expected probability of overload per edge server within a pre-defined range.
The main contributions of this paper are summarized as follows:
\begin{itemize}
  \item We consider a WDTN, where edge servers are responsible for the placement of DT components and performing the associated computations. Considering the limited computational capabilities of edge servers, we formulate the overall cost minimization-oriented placement model as a chance-constrained problem. Therein, a probability function is one of the constraints to ensure network sustainability, i.e., restricting the probability of overload per edge server below a pre-set threshold.
  \item The formulated minimization problem is challenging to solve. The main difficulties stem from its non-convexity and the feasibility checking issue. To overcome the difficulty of evaluating solution feasibility in the original chance-constrained problem, we employ the SAA method to transform the original problem into a deterministic integer linear programming (ILP) formulation. This transformation simplifies the evaluation process and enables us to explore feasible solutions efficiently. Moreover, we demonstrate that a feasible solution obtained from solving the transformed problem has a high probability of also being a feasible solution for the original problem.
  \item In regards to tackling ILP problems, local search-enabled methods are commonly used for efficient solution search. However, the performance of these methods heavily relies on the chosen starting point. Thereby,  initialization is of high necessity to avoid becoming stuck in local optima. To address this issue, we propose an improved local search algorithm. This approach serves as an intellectual restarting local-search algorithm that leverages insights from historical trajectories to learn improved starting points. By doing so, the algorithm aims to find high-quality near-optimal solutions to the transformed ILP problem, enhancing the overall efficiency of the optimization process.

\end{itemize}

At last, extensive numerical simulations are conducted to evaluate the performance of our proposed algorithm. The simulation results show that our developed strategy can converge within several iterations. The results also demonstrate the superiority of our algorithm compared to various baselines in terms of cost savings. Notably, despite achieving superior solutions, the total number of the searched states of our proposed algorithm remains comparable to (or less than) that of the benchmark strategies. These findings further validate the efficacy and efficiency of the devised approach, highlighting its potential for practical implementation in real-world scenarios.
The remainder of this paper is organized as follows: In Section \ref{sec2}, the system model of our considered WDTN is presented, followed by the problem formulation. In Section \ref{sec3}, we elaborate on the problem transformation and the concrete complexity analysis of this transformed problem. With the preliminaries mentioned above, the algorithm design for the sustainability-aware DT placement problem is given in Section \ref{sec4}. In Section \ref{sec5}, numerical results are presented to show the validity of our proposed time-efficient decision-making algorithm. Finally, we summarize this work and predict future research directions in Section \ref{sec6}. The main notations used throughout the paper are summarized in Table I.

\begin{table}
  \footnotesize
  \centering
  \caption{List of notations}
  \begin{tabular}{@{}ll@{}}
    \toprule
    Notation & Definition \\
    \midrule
    $S$                                     & Number of edge servers       \\
    $D$                                     & Number of physical devices    \\
    $m_s$                                   & Required cost per CPU cycle of edge server $s$ \\
    $C_d$                                   & Components number that constitute the DT of physical \hfill\\&device $d$    \\
    $n_c^d$                                 & Required CPU cycles of the $c$-th component in $\cC_d$   \\
    $h_c^d$                                 & Bit size of the $c$-th component in $\cC_d$   \\
    $g_{cc'}^d$                             & Number of bits that need to be exchanged between \hfill\\&components $c$ and $c'\in\cC_d$    \\
    $x_{sc}^{d}$                            & Placement indicator with regard to the $c$-th component \hfill\\&of DT $d$    \\
    $e_{s}^{d}$                             & Manhattan distance between edge server $s$ and physical \hfill\\&device $d$    \\
    $r$                                     & Cost of transmitting one KB of data over one meter of \hfill\\&distance    \\
    $\delta_{sc}^{d}$                       & Cost of offloading for transmitting $c$ of DT $d$ to edge \hfill\\&server $s$    \\
    $l_{ss'}$                               & Manhattan distance between $s$ and $s'\in\cS$    \\
    $\tau_{ss'cc'}^{d}$                     & Cost of communicating between $c$ and $c'$ for DT $d$     \\
    $\gamma_{sc}^{d}$                       & Computation cost produced by the $c$-th component of DT $d$ \hfill\\&at edge server $s$    \\
    $\alpha$                                & Risk factor of the network in $\cP_{1}$   \\
    $\text{Cost}_{\text{off}}$              & Total offloading cost    \\
    $y_{ss'cc'}^{d}$                        & Communication indicator between $c$ and $c'$ for DT $d$    \\
    $\text{Cost}_{\text{com}}$              & Total communicating cost    \\
    $\varrho$                               & Overall cost    \\
    $\bxi$                                  & Index set of observed samples of $\cP_{1}$    \\
    $\Theta$                                & Size of the index set $\bxi$  \\
    $\theta$                                & Observation sample of $\cP_{1}$    \\
    $\widetilde{n}_{c}^{d\theta}$           & Required CPU cycles of the $c$-th component in $\cC_d$ in the \hfill\\&sample $\theta$    \\
    $\widetilde{\gamma}_{sc}^{d\theta}$     & Observed computation cost produced by the $c$-th component \hfill\\&of DT $d$ at edge server $s$ in the sample $\theta$    \\
    $\text{Cost}_{\text{comp},s}^{\theta}$  & Observed total computation cost in sample $\theta$   \\
    $\varepsilon$                           & Risk factor of the network in $\cP_{2}$    \\
    $P^{\Theta}_{\alpha,\varepsilon}$       & Approximating success probability    \\
    $z_{s\theta}$                           & Overload indicator for edge server $s$ in sample $\theta$   \\
    $\bm{a}$                                & Index set of the network state   \\
    $\pi$                                   & Searching policy   \\
    $t$                                     & Iteration times of the searching policy \\
    $\varrho_{t,Q_t}$                       & Local optimum for the searching policy at the $t$-th iteration \\
    $Q_t$                                   & Times that the searching policy converge at $t$-th iteration   \\
    $\Omega_t$                              & Set of data pairs concluding network states and the local \hfill\\&optimum of the $t$-th iteration  \\
    $V$                                     & Mapping between the starting states and their local optima    \\
    $\Delta$                                & Convergence bound for the proposed search algorithm   \\
    \bottomrule
  \end{tabular}
\end{table}

\section{System Model and Problem Formulation}
\label{sec2}

In this section, we first introduce the system model of WDTNs. Then, we formulate the cost minimization-oriented placement problem with the consideration of sustainability control.

\subsection{System Model}

We consider a WDTN that consists of $S$ edge servers, wherein the edge servers are used to assist $D$ physical devices with their DTs realization. Denote by $\cS = \{1,2,\ldots,S\}$ and $\cD= \{1,2,\ldots,D\}$ the index sets of all edge servers and physical devices, respectively. 
Aligning with the settings of practical systems, it is presumed that the edge servers are heterogeneous. More specifically, they have distinct computing costs and computation capacities.
To characterize these heterogeneous features, we use $m_s$ to represent the required cost per CPU cycle of edge server $s$, where $s\in\cS$.
Similarly, let $T_s$ be the computational capability of edge server $s$, which can be interpreted as the maximum cost that edge server $s$ can tolerate for DTs realization in each placement.
However, due to the limited computational capability of each edge server, placing complex DT tasks on a target edge server can lead to heavy workloads and ultimately render the placement unachievable. To tackle this challenge, our model leverages network function visualization orchestration (NFVO) technology \cite{NFVO}, which allows for the separation of each DT into multiple components. More specifically,
for $d\in\cD$, let $C_d$ be the number of the components that constitute the DT of physical device $d$. Denote by $\cC_d = \{1,2,\ldots,C_d\}$ the associated index set. A three-dimensional tuple is then used to capture each component $c\in\cC_d$, referred to as $(n_c^d, h_c^d, g_{cc'}^d)$. Thereof, $n_c^d$ and $h_c^d$ represent the required CPU cycles and the bit size of the $c$-th component in $\cC_d$, respectively. In addition, $g_{cc'}^d$ is the number of bits that need to be exchanged between components $c$ and $c'\in\cC_d$ \cite{placementModel}.

\begin{figure}
    \centering
    \includegraphics[width=2.8in]{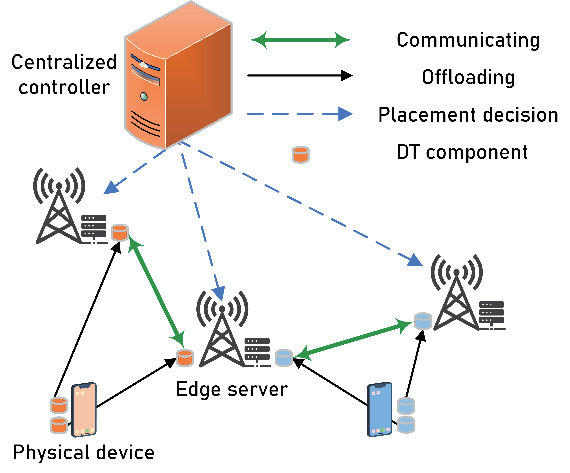}
    \caption{System model of WDTNs.}
    \label{fig1}
\end{figure}

Fig. \ref{fig1} illustrates the detailed DT placement processes in WDTNs, which consist of three steps. First, physical device $d$ generates its DT components, i.e., $\cC_d$, by using NFVO. This step is implemented with the cooperation between independent software vendors and standard edge servers.
Second, DTs of physical devices broadcast their offloading requests to all edge servers. According to these requests, the centralized controller generates the placement strategy by producing a binary vector $\bm{x} = \{\bm{x}_{1},\bm{x}_{2},\ldots,\bm{x}_{d},\ldots,\bm{x}_{D}\}$, where $\bm{x}_{d} = \{x_{sc}^{d}:s \in \cS, c \in \cC_{d}\}$ for $d\in\cD$. Thereof, $x_{sc}^{d} \in \{0,1\}$ is the placement indicator with regard to the $c$-th component of DT $d$.
Specifically, $x_{sc}^{d} =1$ means that the component $c$ of DT $d$ is placed at edge server $s$ and $x_{sc}^{d} =0$ otherwise. In the third step, the components per DT are offloaded, computed, and placed at different edge servers.
Moreover, components that belong to the same DT could be placed at different edge servers. These components need to communicate with each other while preceding their computations.

With the foregoing discussions, we provide a comprehensive lens to look at the inherent costs induced by the aforementioned steps. More precisely, three types of costs are taken into account, which are induced by offloading, communication, and computation, respectively. They are defined as follows:
\subsubsection{Offloading Cost}
We define $e_{s}^{d}$ as the Manhattan distance between edge server $s$ and physical device $d$, and define $r$ as the cost of transmitting one KB of data over one meter of distance. The cost of offloading can be expressed as follows:
\begin{equation}
\delta_{sc}^{d} = e_{s}^{d} \cdot h_{c}^{d} \cdot r.
\label{equCostCompStore}
\end{equation}

\subsubsection{Communication Cost}
As we mentioned above, to prevent placement unachievable caused by heavy workloads brought by complex DT tasks, we separate DT into DT components. Consequently, during the placement process, communication costs arise among these components within the network.
We define $l_{ss'}$ as the Manhattan distance between $s$ and $s'\in\cS$. This communication cost can be expressed as
\begin{equation}
\tau_{ss'cc'}^{d} = l_{ss'} \cdot g_{cc'}^{d} \cdot r.
\label{equCostCompComm}
\end{equation}

\subsubsection{Computation Cost}
The computation cost produced by the $c$-th component of DT $d$ at edge server $s$ is given by:
\begin{equation}
\gamma_{sc}^{d} = m_{s} \cdot n_{c}^{d}.
\label{equCostCompRun}
\end{equation}
It is noteworthy that the status of physical devices is uncertain and cannot be accurately predicted, incurring the high uncertainty of DT placement needs of physical devices. Namely, the required CPU cycles, i.e., $n_{c}^{d}$, become nondeterministic parameters \cite{difficultProblem} with finite distribution. On this basis, $\gamma_{sc}^{d}$ should also be nondeterministic parameters. 
In addition, this type of cost has an impact on network sustainability, which is quantified as the statistical probability of avoiding overload.
To address the sustainability control issue, a chance-constrained formulation is employed, providing a more accurate representation of this scenario \cite{riskAware}. Defining $\alpha$ as the risk factor of the network, which represents an acceptable/pre-determined failure rate of placement, the sustainability control expression can be presented as follows:
\begin{equation}
p\bigg\{\sum^{D}_{d=1}\sum^{C_{d}}_{c=1} \gamma_{sc}^{d} \cdot x_{sc}^{d} \leq T_{s}\bigg\} \geq (1-\alpha), ~\forall s.
\label{equSuccessProb}
\end{equation}



\subsection{Problem Formulation}

In accordance with previous analyses, we now proceed to formulate the optimization problem in this subsection.  The objective of our considered optimization problem is to minimize the overall costs of WDTNs while restricting the placement failure rates caused by exceeding the computation capacities of edge servers within the pre-determined thresholds. To achieve this, we first provide the definition of the overall cost.
Let $\text{Cost}_{\text{off}}$ be the total offloading cost, which can be calculated as
\begin{equation}
\text{Cost}_{\text{off}} = \sum_{d\in\cD}\sum_{s\in\cS}\sum_{c\in\cC_{d}} \delta_{sc}^{d} \cdot x_{sc}^{d}.
\label{equCostOffloading}
\end{equation}
Next, we calculate the corresponding communication cost. For notation simplicity, we introduce a binary variable $y_{ss'cc'}^{d}$ to indicate the occurrence of communication cost between two components. Specifically, $y_{ss'cc'}^{d}=1$ if and only if both variables $x_{sc}^{d}$ and $x_{s'c'}^{d}$ are $1$, and $y_{ss'cc'}^{d}=0$ otherwise. For $d\in\cD$, let $\bm{y} = \{y_{ss'cc'}^{d}:s,s' \in \cS; c,c' \in \cC_{d}\}$ represent the vector of communication indicators. Therefore, the total communication cost, referred to as $\text{Cost}_{\text{com}}$, can be expressed as
\begin{equation}
\text{Cost}_{\text{com}} = \sum_{d\in\cD}\sum_{s\in\cS}\sum_{c\in\cC_{d}}\sum_{s'\in\cS}\sum_{c'\in\cC_{d}} \tau_{ss'cc'}^{d} \cdot y_{ss'cc'}^{d}.
\label{equCostCommunicating}
\end{equation}
The overall cost, which is the sum of the offloading cost and the communication cost,  
denoted as $\varrho$, can be expressed as follows:

\begin{equation}
\varrho=\text{Cost}_{\text{off}}+\text{Cost}_{\text{com}}.
\label{equOptEquitor}
\end{equation}
Based on the given definitions, our minimization problem can be mathematically formulated as follows:
\begin{align}
\cP_{1}:\underset{\bm{x},\bm{y}}{\min}~ & \varrho\nonumber\\
\mathrm{ s.t.}~ & \mathrm{C1:}~ p\bigg\{\sum^{D}_{d=1}\sum^{C_{d}}_{c=1} \gamma_{sc}^{d} \cdot x_{sc}^{d} \leq T_{s}\bigg\} \geq (1-\alpha), ~\forall s,\nonumber\\
& \mathrm{C2:}~ \sum^{S}_{s=1} x_{sc}^{d} = 1, ~\forall c,d,\nonumber\\
& \mathrm{C3:}~ x_{sc}^{d} \geq y_{ss'cc'}^{d}, ~\forall s,s',c,c',d,\nonumber\\
& \mathrm{C4:}~ x_{s'c'}^{d} \geq y_{ss'cc'}^{d}, ~\forall s,s',c,c',d,\nonumber\\
& \mathrm{C5:}~ x_{sc}^{d} + x_{s'c'}^{d} - 1 \leq y_{ss'cc'}^{d}, ~\forall s,s',c,c',d,\nonumber\\
& \mathrm{C6:}~ x_{sc}^{d}\in\{0,1\}, ~\forall s,c,d,\nonumber\\
& \mathrm{C7:}~ y_{ss'cc'}^{d}\in\{0,1\}, ~\forall s,s',c,c',d.\nonumber
\end{align}
C1 states that the probability of surpassing the computation capacity per edge server must not exceed the risk level $\alpha$. C2 restricts that any DT component can only be placed once. C3, C4, and C5 ensure that when placements of components are made, the communication indicators are set accordingly. In addition, C6 and C7 represent the integrity requirements for the decision variables.

As shown above, the formulated optimization problem involves nondeterministic parameters, which need to satisfy a specified probability threshold, and this problem is difficult to solve.
Specifically, the difficulties primarily stem from three aspects. Firstly, the high uncertainty surrounding $n_{c}^{d}$ results in different edge servers facing varying probability distributions for the required CPU cycles. Addressing this issue necessitates multidimensional integration across these diverse probability distributions corresponding to $s\in\cS$ \cite{SAA_Analysis}.
Moreover, when dealing with nonclassical and complex probability distributions, estimating detailed distributions through methods related to cumulative distribution functions (CDF) becomes impractical  \cite{CDF_Approximation}.
Secondly, checking the feasibility of a given possible solution is impossible. Thirdly, the feasible region induced by constraints is non-convex.
To achieve this goal, approximating methods can be used to implement problem transformation  \cite{ScenarioApproach,RandomizedApproach,difficultProblem}.
Among the methods for solving chance-constrained optimization problems, the SAA method has gained significant attention \cite{SAAintroduce}. It is commonly employed to estimate the expectation of a stochastic program using Monte Carlo simulation-based approaches. Compared with other methods, SAA can offer improved candidate solutions in the chance-constrained problem. We adopt this method in our work due to the aforementioned reasons. In the next section, we present an approximating problem of $\cP_{1}$, which is transformed by the SAA method. Then, a concrete complexity analysis for this transformed problem is sorted out.

\section{Problem Transformation and Complexity Analysis}
\label{sec3}

In this section, we first transform the original sustainability-aware problem into an ILP by the SAA method. Then, we provide a rigorous complexity analysis of the transformed problem.

\subsection{Problem Transformation}
In this subsection, we will present the transformed problem of $\cP_{1}$ using the SAA method. To facilitate the analysis, we define $\bxi = \{1,2,\ldots,\Theta\}$ as an index set of observed samples of $\cP_{1}$.
In different samples, components of DT $d$ have different required CPU cycles for its placement, which is denoted as $n_{c}^{d}$. To differentiate $n_{c}^{d}$ in different samples, we rewrite it as $\widetilde{n}_{c}^{d\theta}$, where $\theta\in\bxi$. This notation means the required CPU cycles of $c$ in DT $d$ are observed in the sample $\theta$. Based on the SAA method \cite{SAA_Analysis}, $\{\widetilde{n}_{c}^{d\theta}:\theta\in\bxi\}$ is an independent Monte Carlo sample set of nondeterministic parameter $n_{c}^{d}$. In reality, these parameters are generally obtained from historical data. Correspondingly, we rewrite (\ref{equCostCompRun}) as
\begin{equation}
\widetilde{\gamma}_{sc}^{d\theta} = m_{s} \cdot \widetilde{n}_{c}^{d\theta},~\forall s,c,d,\theta,
\label{equCostCompRunApproxi}
\end{equation}
and refer to it as computation cost produced by the $c$-th component of DT $d$ at edge server $s$ observed in sample $\theta$. It should be noted that the values of $m_{s}$ and $T_{s}$ ($s\in\cS$) remain unchanged in every sample. That is because we consider edge servers to be fixed, as aforementioned. Accordingly, these properties remain deterministic. Similarly, we define $\text{Cost}_{\text{comp},s}^{\theta}$ as computation cost of server $s$ in sample $\theta$, it can be calculated as
\begin{equation}
\text{Cost}_{\text{comp},s}^{\theta} = \sum_{d=1}^{D}\sum_{c=1}^{C_{d}}\widetilde{\gamma}_{sc}^{d\theta}\cdot x_{sc}^{d},~\forall s,\theta.
\label{equCostComputation}
\end{equation}
Then, we compare this observed cost with the computational capacity of $s$, and term this difference value as the remaining computational capability of $s$. We write it as:
\begin{equation}
G_{s}(\bm{x},\theta) = \text{Cost}_{\text{comp},s}^{\theta} - T_{s}, ~s\in\cS,\theta\in\bxi.
\label{equAppEnergy}
\end{equation}
If $G_{s}(\bm{x},\theta)$ is larger than 0, the computation load on edge server $s$ in sample $\theta$ exceeds its computational capacity. In this case, we can say that an overload occurs at edge server $s$ in sample $\theta$. On the other hand, if $G_{s}(\bm{x},\theta)$ is equal to or less than 0, it means that the computation load lies within the CPU budget of edge server $s$, and an overload does not occur.
In other words, $G_{s}(\bm{x},\theta)$ can effectively indicate the overloads of $s\in\cS$ in sample $\theta\in\bxi$. For clarity, we define an indicator function $\mathds{1}(G_{s}(\bm{x},\theta))$, which takes value of 1 if $G_{s}(\bm{x},\theta)> 0$ and 0 otherwise. Further, we define
\begin{equation}
P_{s}(\bm{x}) = \frac{1}{\Theta}\sum_{\theta=1}^{\Theta}\mathds{1}(G_{s}(\bm{x},\theta)), ~s\in\cS
\label{equSAAproportion}
\end{equation}
as the proportion of realizations in the samples where $G_{s}(\bm{x},\theta) > 0$, indicating an overload at edge server $s$.
However, it is important to note that this proportion is obtained from observations, which cannot be used to compare with $\alpha$ directly. To mitigate this issue, we give another risk level $\varepsilon\in(0,\alpha]$, and express the sample approximation of sustainability control as
\begin{equation}
P_{s}(\bm{x}) \leq \varepsilon, ~\forall s.
\label{equSAAexpression}
\end{equation}
This means that for each $s\in\cS$, the proportion of realizations of overload in samples is less than or equal to $\varepsilon$, which is no larger than $\alpha$. In the paper authored by S. Ahmed \emph{et al.} \cite{difficultProblem}, they have proven that if we replace C1 of $\cP_{1}$ with (\ref{equSAAexpression}), the solution of the transformed problem can be a lower bound of the solution of $\cP_{1}$ with a high probability.
In fact, this probability, which can be termed as approximating success probability, is calculated by
\begin{equation}
P^{\Theta}_{\alpha,\varepsilon} = 1-\exp{\left(-\Theta\frac{(\alpha-\varepsilon)^2}{2\varepsilon}\right)}.
\label{equFeasibleSolu}
\end{equation}
To increase the likelihood of $P^{\Theta}_{\alpha,\varepsilon}$ approaching 1, we require a large value for $\Theta$ and a suitably small value for $\varepsilon$ in our approximation.
In this manner, we can convert the chance-constrained problem $\cP_{1}$ into a problem that is subject to a determined proportion. We can also check the feasibility of the solution by examining the sample realizations \cite{difficultProblem}. As a consequence, we can overcome the difficulty of feasibility checking in $\cP_{1}$ as discussed at the end of the previous section. However, it is still unclear how to adjust the solution to satisfy this proportion. Therefore, we rewrite the approximating problem as an ILP by introducing a new set of binary variables. Specifically, we define the vector of variables as $\bm{z}=\{z_{s\theta}:s\in\cS,\theta\in\bxi\}$, where, $z_{s\theta}\in\{0,1\}$ indicates if edge server $s$ is overloaded in sample $\theta$. Specifically, $z_{s\theta}=1$ means that the computational capacity of edge server $s$, referred to as $T_{s}$, is less than its computation cost in sample $\theta$, and $z_{s\theta}=0$ otherwise. Since all of control variables are vectors, in which the elements are binary variables, we can combine them by defining a state as $\bm{a} \triangleq [\bm{x}, \bm{y}, \bm{z}]^{T}$ for the conciseness, where $\bm{x} = [x_{11}^{1},\ldots,x_{sc}^{d},\ldots,x_{SC_{D}}^{D}]$, $\bm{y} = [y_{1111}^{1},\ldots,y_{ss'cc'}^{d},\ldots,y_{SSC_{D}C_{D}}^{D}]$, and $\bm{z} = [z_{11},\ldots,z_{s\theta},\ldots,z_{S\Theta}]$. Correspondingly, the set of all possible states can be denoted as $\cA$.
Based on the analyses above, $\cP_{1}$ can be reformulated as the following ILP:
\begin{align}
\cP_{2}:\underset{\bm{a}}{\min}~ & \varrho\nonumber\\
\mathrm{ s.t.}~ & \mathrm{C1:}~ G_{s}(\bm{x},\theta) \leq W\cdot z_{s\theta}, ~\forall s,\theta,\nonumber\\
& \mathrm{C2:}~ \sum_{\theta=1}^{\Theta}z_{s\theta}\leq\varepsilon\cdot\Theta, ~\forall s,\nonumber\\
& \mathrm{C3:}~ \sum^{S}_{s=1} x_{sc}^{d} = 1, ~\forall c,d,\nonumber\\
& \mathrm{C4:}~ x_{sc}^{d} \geq y_{ss'cc'}^{d}, ~\forall s,s',c,c',d,\nonumber\\
& \mathrm{C5:}~ x_{s'c'}^{d} \geq y_{ss'cc'}^{d}, ~\forall s,s',c,c',d,\nonumber\\
& \mathrm{C6:}~ x_{sc}^{d} + x_{s'c'}^{d} - 1 \leq y_{ss'cc'}^{d}, ~\forall s,s',c,c',d,\nonumber\\
& \mathrm{C7:}~ x_{sc}^{d}\in\{0,1\}, ~\forall s,c,d,\nonumber\\
& \mathrm{C8:}~ y_{ss'cc'}^{d}\in\{0,1\}, ~\forall s,s',c,c',d,\nonumber\\
& \mathrm{C9:}~ z_{s\theta}\in\{0,1\}, ~\forall s,\theta,\nonumber
\end{align}
where $W$ is a very large positive integer. In $\cP_{2}$, C1 illustrates whether the constraint of computation capacity of $s$ in sample $\theta$ is violated. Meanwhile, C2 ensures that the proportion of overload is not greater than the risk level $\varepsilon$ in the ILP, i.e., (\ref{equSAAexpression}). C3-C8 are identical to the corresponding constraints of the original problem $\cP_{1}$. C9 represents the integrity requirements for the binary vector $\bm{z}$.
According to the analysis of the SAA method in \cite{difficultProblem,SAA_Analysis}, for $\varepsilon<\alpha$, as $\Theta$ increases, the likelihood that feasible solutions of $\cP_{2}$ will also be feasible for $\cP_{1}$ will be high.
In the next subsection, we will provide the time complexity analysis of this transformed problem.

\subsection{Complexity Analysis} \label{complexity}
\label{sec3b}

In a WDTN system, the time efficiency of the placement problem is crucial due to the limited CPU resources available. However, this requirement conflicts with the non-convex nature of the feasible region. The complexity analysis of $\cP_{2}$ emphasizes the critical nature of this challenge.
\begin{myth}
$\cP_{2}$ is NP-hard.
\end{myth}
\begin{IEEEproof}
We prove $\cP_{2}$ is NP-hard by showing that there is a polynomial-time reduction that maps any instance of an NP-hard problem to this problem. Here, we consider an NP-hard problem called multi-component application placement problem (MCAPP), which is a mixed integer non-linear programming (MINLP) problem and has been proven to be NP-hard in \cite{NPhard}. It is expressed as follows:
\begin{align}
\text{MCA}&\text{PP-MINLP}:\nonumber\\
\min~ & \sum_{i=1}^{m}\sum_{j=1}^{n}(\omega_{ij}\cdot x_{ij}+\sum_{i'=1}^{m}\sum_{j'=1}^{n}\tau_{ii'jj'}\cdot x_{ij}\cdot x_{i'j'})\nonumber\\
\mathrm{ s.t.}~ & \mathrm{C1:}~ \sum_{j=1}^{n}x_{ij}\leq 1,~i=1,\ldots,m,\nonumber\\
& \mathrm{C2:}~ \sum_{i=1}^{m}x_{ij}=1,~j=1,\ldots,n,\nonumber\\
& \mathrm{C3:}~ x_{ij}\in\{0,1\},~i=1,\ldots,m;~j=1,\ldots,n,\nonumber
\end{align}
where $m$ is the number of servers and $n$ is the number of components. Besides, $\omega_{ij}$ is the server-component cost, and $\tau_{ii'jj'}$ is the inter-component cost. The decision variable $x_{ij}$ is set to 1 if component $j$ is assigned to server $i$ and 0 otherwise. Given an instance of MCAPP-MINLP with the aforementioned parameters, we construct an instance of $\cP_{2}$ with 1 physical device, $m$ servers, and $n$ components. For each server $i\in\{1,2,\ldots,m\}$ and component $j\in\{1,2,\ldots,n\}$, we set: $\varepsilon\triangleq1$, $\delta_{ij}\triangleq\omega_{ij}$, and $\widetilde{n}_{ij}^{\theta}\triangleq\frac{W+T_{i}}{m_{i}}$. Thus, the overall cost can be written as
\begin{align}
&\varrho' = \sum_{i=1}^{m}\sum_{j=1}^{n} \omega_{ij} \cdot x_{ij}+\nonumber\\
& \sum_{i=1}^{m}\sum_{j=1}^{n}\sum_{i'=1}^{m}\sum_{j'=1}^{n} \tau_{ii'jj'} \cdot y_{ii'jj'},
\label{equProofOverCost}
\end{align}
and $\cP_{2}$ can be written as
\begin{align}
\cP'_{2}:\underset{\bm{a}}{\min}~ & \varrho'\nonumber\\
\mathrm{ s.t.}~ & \mathrm{C1:}~ \sum_{j=1}^{m}(W+T_{i})\cdot x_{ij}-T_{i}\leq W\cdot z_{i\theta},~\forall i,\theta,\nonumber\\
& \mathrm{C2:}~ \sum_{\theta=1}^{\Theta}z_{i\theta}\leq\Theta, ~\forall s,\nonumber\\
& \mathrm{C3:}~ \sum^{m}_{i=1} x_{ij} = 1, ~\forall j,\nonumber\\
& \mathrm{C4:}~ x_{ij} \geq y_{ii'jj'}, ~\forall i,i',j,j',\nonumber\\
& \mathrm{C5:}~ x_{ij} + x_{i'j'} - 1 \leq y_{ii'jj'}, ~\forall i,i',j,j',\nonumber\\
& \mathrm{C6:}~ x_{ij}\in\{0,1\}, ~\forall i,j,\nonumber\\
& \mathrm{C7:}~ y_{ii'jj'}\in\{0,1\}, ~\forall i,i',j,j',\nonumber\\
& \mathrm{C8:}~ z_{i\theta}\in\{0,1\}, ~\forall i,\theta.\nonumber
\end{align}
It is noticed that $z_{i\theta}$ is a decision variable that would not affect the result in $\cP'_{2}$. Therefore, we let $z_{i\theta} = 1,~\forall i,\theta$. Then, the problem can be written as
\begin{align}
\cP''_{2}:\underset{\bm{x},\bm{y}}{\min}~ & \varrho'\nonumber\\
\mathrm{ s.t.}~ & \mathrm{C1:}~ \sum_{j=1}^{m}x_{ij}\leq 1,~\forall i,\theta,\nonumber\\
& \mathrm{C2:}~ \sum^{m}_{i=1} x_{ij} = 1, ~\forall j,\nonumber\\
& \mathrm{C3:}~ x_{ij} \geq y_{ii'jj'}, ~\forall i,i',j,j',\nonumber\\
& \mathrm{C4:}~ x_{ij} + x_{i'j'} - 1 \leq y_{ii'jj'}, ~\forall i,i',j,j',\nonumber\\
& \mathrm{C5:}~ x_{ij}\in\{0,1\}, ~\forall i,j,\nonumber\\
& \mathrm{C6:}~ y_{ii'jj'}\in\{0,1\}, ~\forall i,i',j,j'.\nonumber
\end{align}
In general, linear programming can not be converted to non-linear programming. However, in some specific situations, i.e. integer programming, we can transform one to another by adding or removing optimizing variables. From C3 in $\cP''_{2}$, we know that if $x_{ij}=0$ or $x_{i'j'}=0$, then $y_{ii'jj'}=0$. From C4 in $\cP''_{2}$, we know that if $x_{ij}=1$ and $x_{i'j'}=1$, then $y_{ii'jj'}=1$. Thus, we can express $y_{ii'jj'}=x_{ij}\cdot x_{i'j'}, ~\forall i,i',j,j'$.  Thus, $\cP''_{2}$ can be reduced to:
\begin{align}
\cP'''_{2}:\underset{\bm{x}}\min~ & \sum_{i=1}^{m}\sum_{j=1}^{n}(\omega_{ij}\cdot x_{ij}+\sum_{i'=1}^{m}\sum_{j'=1}^{n}\tau_{ii'jj'}\cdot x_{ij}\cdot x_{i'j'})\nonumber\\
\mathrm{ s.t.}~ & \mathrm{C1:}~ \sum_{j=1}^{n}x_{ij}\leq 1,~i=1,\ldots,m,\nonumber\\
& \mathrm{C2:}~ \sum_{i=1}^{m}x_{ij}=1,~j=1,\ldots,n,\nonumber\\
& \mathrm{C3:}~ x_{ij}\in\{0,1\},~i=1,\ldots,m;~j=1,\ldots,n.\nonumber
\end{align}
It is noticed that $\cP'''_{2}$ is identical with MCAPP-MINLP. Hence, we have a polynomial-time reduction from any instance of the MCAPP-MINLP to an instance of $\cP_{2}$. By reducing a decision version of MCAPP into the traveling salesman problem (TSP), T. Bahreini et \textsl{al} \cite{NPhard} have proved MCAPP-MINLP as NP-hard. Therefore, $\cP_{2}$ is NP-hard.
\end{IEEEproof}

Building on these results, in Section \ref{sec4}, we introduce an improved local search algorithm to efficiently discover suboptimal solutions for $\cP_{2}$. Furthermore, in Section \ref{sec5}, we compare the achieved suboptimal solutions with those found by several baselines to demonstrate the superiority of our proposed algorithm.

\section{Algorithm Design}
\label{sec4}
In this section, we introduce the proposed placement algorithm, which is called the improved local search algorithm, to solve $\cP_{2}$. We will first provide an overview of the developed scheme and then elaborate on the details of the algorithm.

\subsection{Overview of the Proposed Algorithm} \label{over}
Based on the analyses in Section \ref{complexity}, it is evident that solving problem $\cP_{2}$ optimally in polynomial time is not feasible. However, local search algorithms can efficiently search for local optimal solutions. It is worth noting that the quality of local search methods is highly dependent on the selection of the starting state/point. Yet, when the search space is both extremely large and non-convex, finding the starting state that leads to the global optimal state can be challenging. To address this problem, we propose an improved local search algorithm by exploring the relationship between starting states and corresponding local optima. More specifically, the proposed scheme works in an iterative manner, where each iteration consists of two phases. In the first phase, we start from a given starting state generated by the previous iteration and use conventional local search methods to reach a local optimal solution to $\cP_{2}$. Note that, for a given specific local search method, the trajectory of its convergence is predetermined. Each point on the trajectory corresponds to paired data, i.e., the system state and the associated local optimal value of $\cP_{2}$ under this state. Moreover, taking any point on this trajectory as the starting point will lead to the same local optimal solution \cite{STAGEPhD}. These points will be used in phase two, where we use the generated paired data from phase one to learn the mapping relationship between the starting states and their local optima. After obtaining the mapping function, we can apply the local search method used in phase one again to produce a good initial point, which will be used as the starting state of phase one in the next iteration.

Based on the foregoing discussions, it can be observed that our developed algorithm functions as a smart-restart approach to conventional local search methods, which has the potential to significantly enhance the results obtained through local search. 
For brevity, we summarize the mechanism of the $t$-th iteration for our proposed search algorithm in Fig. \ref{fig2}.
In the following section, we will provide more details about the two phases involved in each iteration.

\begin{figure}
    \centering
    \includegraphics[width=3.2in]{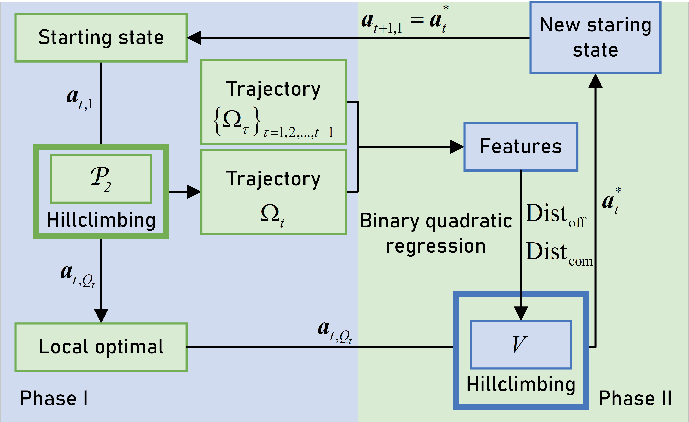}
    \caption{The $t$-th iteration of our improved local search algorithm.}
    \label{fig2}
\end{figure}

\subsection{Details of the Proposed Algorithm}
Our designed algorithm works in an iterative manner. Without loss of generality, we take the $t$-th iteration as an example to introduce the detailed implementations of each phase. For simplicity, we make some definitions that will be used in the following analyses. We define the starting state of the $t$-th iteration as $\bm{a}_{t,1}\in\cA$, and denote its corresponding objective value by $\varrho_{t,1}$, which is calculated by (\ref{equOptEquitor}). With a given searching policy $\pi$, starting from $\bm{a}_{t,1}$, a local optimal state can be reached. It is noteworthy that, in this context, the symbol $\pi$ represents any local search method. We assume that $\pi$ requires $Q_t$ times of searching to converge.
Thus, the local optimal state can be represented by $\bm{a}_{t,Q_t}$, as shown in the left side of Fig. \ref{fig2}. Its associated objective value is the local optimum of this iteration, which is given by $\varrho_{t,Q_t}$ based on our definitions.
It can be observed that the $t$-th iteration generates $Q_t$ pairs of data denoted by $\Omega_t=\{(\bm{a}_{t,q}, \varrho_{t,Q_t})_{q=1,2,\ldots,Q_t}\}$. Each data pair can be viewed as a tuple comprising the state vector $\bm{a}_{t,q}$, where $q\in\{1,2,\ldots,Q_t\}$, and the corresponding local optimal value $\varrho_{t,Q_t}$. With the foregoing discussions, we present our improved local search method by introducing the two phases of the algorithm.

\subsubsection{Phase I of the $t$-th iteration}

In the first phase of the $t$-th iteration, we begin with $\bm{a}_{t,1}$ and apply the local search method $\pi$ to find $\bm{a}_{t,Q_t}$, generating $\Omega_t$. This dataset is then used for relationship learning in the second phase. For this study, we utilize the Hillclimbing algorithm as the searching policy ($\pi$) due to its fast convergence, which allows for quicker collection of trajectory samples \cite{whyHillclimbing}. Furthermore, the aforementioned relationship is a mapping between the starting states and their local optima, which can be defined as $V:\cA\rightarrow\mathbb{R}$.
\subsubsection{Phase II of the $t$-th iteration}

In Phase II of the proposed algorithm, we use the states set $\Omega_t$ generated from phase one to learn a closed form of $V$, and predict a better-starting state to enhance the performance of the search results, as shown in the middle of Fig. \ref{fig2}. Specifically, due to the large number of state parameters, we generate the closed form of $V$ by integrating two feature values $\text{Dist}_{\text{off}}$ and $\text{Dist}_{\text{com}}$ from states set $\Omega_t$ by
\begin{equation}
\text{Dist}_{\text{off}} = \sum_{d\in\cD}\sum_{s\in\cS}\sum_{c\in\cC_{d}} e_{sc}^{d} \cdot x_{sc}^{d}
\label{equFeature1}
\end{equation}
and
\begin{equation}
\text{Dist}_{\text{com}} = \sum_{d\in\cD}\sum_{s\in\cS}\sum_{c\in\cC_{d}}\sum_{s'\in\cS}\sum_{c'\in\cC_{d}} l_{ss'} \cdot y_{ss'cc'}^{d},
\label{equFeature2}
\end{equation}
where $x_{sc}^{d},y_{ss'cc'}^{d}\in\bm{a}_{t,q}$,
in order to simplify the fitting complexity and enhancing convergence speed.
The reason for selecting $\text{Dist}_{\text{off}}$ and $\text{Dist}_{\text{com}}$ as the features are the observation of a positive correlation between distance and cost in $\cP_{2}$.
Then, we use a binary quadratic regression model with feature values ($\text{Dist}_{\text{off}}$ and $\text{Dist}_{\text{com}}$) as observations.
It is worth noting that not only the trajectory in the $t$-th iteration but also those from all previous iterations are used to learn $V$.
Furthermore, any state in these trajectories can be considered as a starting state for learning this model, as mentioned in the overview. The binary quadratic regression model provides an approximation of the relationship between starting states and their corresponding local optima.
Therefore, in the regression procedure, we combine $\varrho_{t,Q_t}$, $\text{Cost}_{\text{off}}$, and $\text{Cost}_{\text{com}}$ as a tuple, and then apply the least squares method to perform the regression.
Our algorithm is expected to predict an improved starting state. Based on the learned mapping relationship ($V$), the algorithm takes the local optimal state of phase one, i.e. $\bm{a}_{t,Q_t}$, as the starting point to keep searching the local optimal state of $V$ by using Hillclimbing as well, as shown in the right side of Fig. \ref{fig2}. The newly reached local optimal state, denoted by $\bm{a}_{t}^{*}$, represents a predicted better-starting state compared to the previous starting states ($\bm{a}_{t,1}$) in the trajectory points.
It is not difficult to observe that the prediction accuracy increases along with the iteration numbers.
We then use the predicted starting state from phase two as the initial point of the next iteration, i.e.,
\begin{equation}
\bm{a}_{t+1,1} = \bm{a}_{t}^{*}.
\label{equStartingState}
\end{equation}
For the first iteration of our algorithm, no predicted starting state is available. In this case, the starting state for phase one is generated randomly.

Now, it becomes evident that finding the global optimal state can be difficult. As a result, we define $\Delta$ as the convergence bound for the proposed search algorithm to find a near-optimal state. We then provide a brief overview of the convergence criterion, which is expressed as follows \cite{convergenceSitu}:
\begin{equation}
\frac{|\varrho_{t, Q_{t}}-\varrho_{t-1, Q_{t-1}}|}{|\varrho_{t, Q_{t}}|+|\varrho_{t-1, Q_{t-1}}|}<\Delta,
\label{equFinalConv}
\end{equation}
where (\ref{equFinalConv}) holds true depicts that the state $\bm{a}_{t,Q_t}$ is the near-optimal state the developed algorithm has found. For the sake of clarity, we summarize the pseudo-code of our developed placement method in Algorithm \ref{alg1}.
\begin{algorithm}
	\caption{WDTN placement algorithm.}
	\label{alg1}
	\begin{algorithmic}[1]
        \STATE Let $t = 0$, and define convergence bound as $\Delta$.
        \STATE Initialize $\varrho_{t,Q_{t}}=\infty$.
        \STATE Randomly initialize starting state $\bm{a}_{1,1}$.
        \REPEAT
            \STATE Run Hillclimbing method $\pi$ from $\bm{a}_{t,1}$ on $\cP_{2}$ to find $\bm{a}_{t,Q_t}$ and $\varrho_{t, Q_t}$.
            \STATE Produce trajectory $\Omega_t$.
            \STATE Learn $V$ using binary quadratic regression with $\text{Dist}_{\text{off}}$ and $\text{Dist}_{\text{com}}$ as observations.
            \STATE Run Hillclimbing method $\pi$ from $\bm{a}_{t,Q_t}$ on $V$ to find $\bm{a}_{t}^{*}$.
            \STATE Set $\bm{a}_{t+1,1}$ using (\ref{equStartingState}).
            \STATE Update $t \leftarrow t+1$.
        \UNTIL (\ref{equFinalConv}) becomes true
        \RETURN the near-optimal state $\bm{a}_{t,Q_t}$.
	\end{algorithmic}
\end{algorithm}

It is noteworthy that the formulated problem is challenging to solve, as demonstrated in Section \ref{sec3b}. Although our proposed algorithm encounters difficulties in achieving the optimal solution, it provides an applicable solution for practical systems, particularly in large-scale scenarios.
Before ending this section, we give the complexity analysis of our proposed placement algorithm.
\begin{myle}
The time complexity of Algorithm \ref{alg1} is $\cO(\sum_{t=1}^{T}(Q_{t}|\bm{a}|+(\sum_{i=1}^{t}Q_{i})^{3}))$.
\end{myle}
\begin{proof}
Based on the foregoing analyses, we know that Algorithm \ref{alg1} is constituted by the Hillclimbing method and binary quadratic regression. For its $t$-th iteration, Hillclimbing is used twice to find local optimal states, resulting in a time complexity of $\cO(Q_{t}|\bm{a}|)$ \cite{whyHillclimbing}. In the meantime, a time complexity of $\cO((\sum_{i=1}^{t}Q_{i})^{3})$ is required for binary quadratic regression \cite{regressionComplexity}.
Let $T$ denote the number of iterations required for Algorithm \ref{alg1} to converge, the proof is completed.

\end{proof}

\section{Numerical Results}
\label{sec5}
Extensive numerical simulations are conducted in this section to validate the performance of our proposed algorithm. It is worth noting that the development of WDTNs is still in its early stages; there are currently no publicly available workload traces for WDTNs.
On this basis, we rely on synthetic instances generated to address the DT placement problem in our simulations.
We consider a WDTN that is deployed in a two-dimensional $120\times120~m^{2}$ area, where a specific number of edge servers and physical devices are uniformly distributed throughout this region \cite{AreaSetting}. To demonstrate the performance of the algorithm under different settings, we consider two types of scenarios with varying densities of physical entities, including physical devices and edge servers. In these scenarios, we independently vary the number of edge servers and physical devices to analyze their individual impact on the algorithm's performance.
More specifically, in the edge server changing scenarios, we examine in total of 9 independent cases where the number of edge servers, denoted by $S$, is selected from the set $\{2,3,4,5,6,7,8,9,10\}$. Moreover, two types of physical device numbers are employed: 5 physical devices and 10 physical devices, denoted as $D = 5$ and $D = 10$, respectively.
Furthermore, the number of DT components per physical device ranges from $1$ to $3$, namely, $C_{d}\in[1,3]$, where $d\in\cD$. Meanwhile, in the physical device changing scenarios, we investigate in total of 6 independent cases characterized by the number of physical devices chosen from the set $\{5,6,7,8,9,10\}$. Therein, the number of edge servers is fixed at 6. Accordingly, we consider two different ranges for the DT components, i.e., $[1,3]$ and $[1,5]$ \cite{placementModel}. Regarding the detailed features of edge servers in the foregoing mentioned scenarios, the required cost per CPU cycle $m_s$, as defined in Section \ref{sec2}, follows a normal distribution $N(\mu_s, 0.2\mu_s)$, where $\mu_s$ is uniformly distributed in the range $[1,10]$ \cite{NPhard}. The computational capacity $T_{s}$ is uniformly distributed in the range $[0.3,0.4]\times 10^{9}$ \cite{TsSetting}. Likewise, for the detailed features of physical devices, the required CPU cycles $n_{c}^{d}$ follow a normal distribution $N(\mu_{d},0.2\mu_{d})$, where $\mu_{d}$ is uniformly distributed in the range $[1, 10] \times 10^6$ \cite{NPhard}. Additionally, for each component, the bit sizes for offloading and communication, i.e., $h_{c}^{d}$ and $g_{cc'}^{d}$, are uniformly distributed within the ranges of $[100, 500]$ KB and $[50, 250]$ KB, respectively \cite{TsSetting}. The unit cost for offloading $r$ is uniformly distributed within the range $[0, 1]$. In the meantime, the convergence bound $\Delta$ is set to $0.015$, the risk factor of $\cP_{1}$ is set to $\alpha = 0.01$, and the approximating success probability given by (\ref{equFeasibleSolu}) is $P^{\Theta}_{\alpha,\varepsilon}=0.99$ \cite{difficultProblem}. Based on (\ref{equFeasibleSolu}), the risk factor of $\cP_{2}$ is set to $\varepsilon = 0.005$, and the observation number of the sampling is given by $\Theta = 1850$. Last but not least, we set $T=10$ as the maximum iteration number of our developed method. The key parameters used throughout the simulation are summarized in Table \ref{Tab2}.
\begin{table}
    \footnotesize

    \centering
    \caption{Parameters of placement of WDTN}
   \scalebox{0.9}{ \begin{tabular}{|c|c|c|}
        \hline
        {\bf Parameters}&\multicolumn{1}{c|}{\bf numerical values}\\ \hline
        The required cost per CPU cycle $m_{s}$                     & $N(\mu_{s},0.2\mu_{s})$       \\ \hline
        The base parameter of required costs per CPU cycle $\mu_{s}$              & $U[1,10]$                     \\ \hline
        The computational capability $T_{s}$                        & $U[0.3,0.4]\times 10^{9}$    \\ \hline
        The required CPU cycles $n_{c}^{d}$                         & $N(\mu_{d},0.2\mu_{d})$       \\ \hline
        The base parameter of required CPU cycles $\mu_{d}$         & $U[1, 10] \times 10^6$        \\ \hline
        The bit size of components offloading $h_{c}^{d}$           & $U[100, 500]$                \\ \hline
        The bit size should be exchanged between components $g_{cc'}^{d}$      & $U[50, 250]$                 \\ \hline
        The unit cost for offloading $r$   & $U[0, 1]$                     \\ \hline
        The risk factor of the problem $\alpha$                     & $0.01$                        \\ \hline
        The convergence bound $\Delta$                              & $0.015$                       \\ \hline
        The risk factor of the transformed problem $\varepsilon$    & $0.005$                       \\ \hline
    \end{tabular}}
\label{Tab2}
\end{table}

For performance comparison, the following baselines are considered:
\begin{itemize}
  \item \textbf{Baseline 1}: 
  In this method, we generate $T$ randomly selected states that satisfy all constraints. Among these states, the solution is determined by selecting the state that results in the minimum overall cost.
  \item \textbf{Baseline 2}:
  In this strategy, $T$ randomly selected states are optimized using the Hillclimbing algorithm. The state that incurs the minimum overall cost is considered as the solution.
  \item \textbf{Baseline 3}:
  In this scheme, the placement of DTs is designed to be as close to their physical devices as possible while still satisfying the given constraints.
\end{itemize}

In the following, we consider three performance metrics: the convergence of our proposed algorithm, the average cost per edge server, and the average number of searched states.

\subsection{Convergence Performance}

In this subsection, we evaluate the convergence performance of our improved local search algorithm under various parameter settings. We take the average cost per edge server during each iteration as the evaluation metric, as depicted in Fig. \ref{fig3} and Fig. \ref{fig4}.
Specifically, Fig. \ref{fig3} demonstrates the convergence performance of scenarios where the number of components per physical device ranges from 1 to 3, while the range of physical device numbers varies from 5 to 10. In these scenarios, the number of edge servers remains constant at 6.
On the other hand, Fig. \ref{fig4} presents the performance in scenarios where the number of components per physical device ranges from 1 to 5. Other parameter settings in Fig. \ref{fig4} are the same as those in Fig. \ref{fig3}.
Both Fig. \ref{fig3} and Fig. \ref{fig4} demonstrate that our proposed algorithm achieves convergence within a few iterations, highlighting its efficiency and suitability for real-world systems. As expected, an increase in the number of physical devices leads to higher average cost per edge server. Furthermore, it is noteworthy that our developed strategy accepts increased average costs, allowing it to surpass local optimal solutions that may have inferior performance. The superiority of our scheme over various baselines will be presented minutely in the subsequent sections.

\begin{figure}
    \centering
    \includegraphics[width=3in]{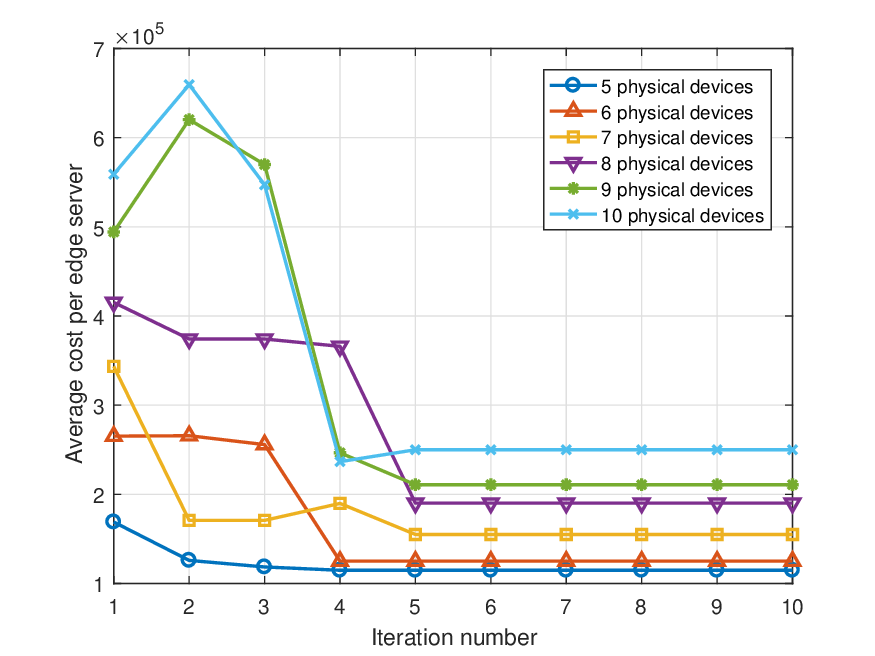}
    \caption{Convergence of average cost with a range of 1-3 components.}
    \label{fig3}
\end{figure}

\begin{figure}
    \centering
    \includegraphics[width=3in]{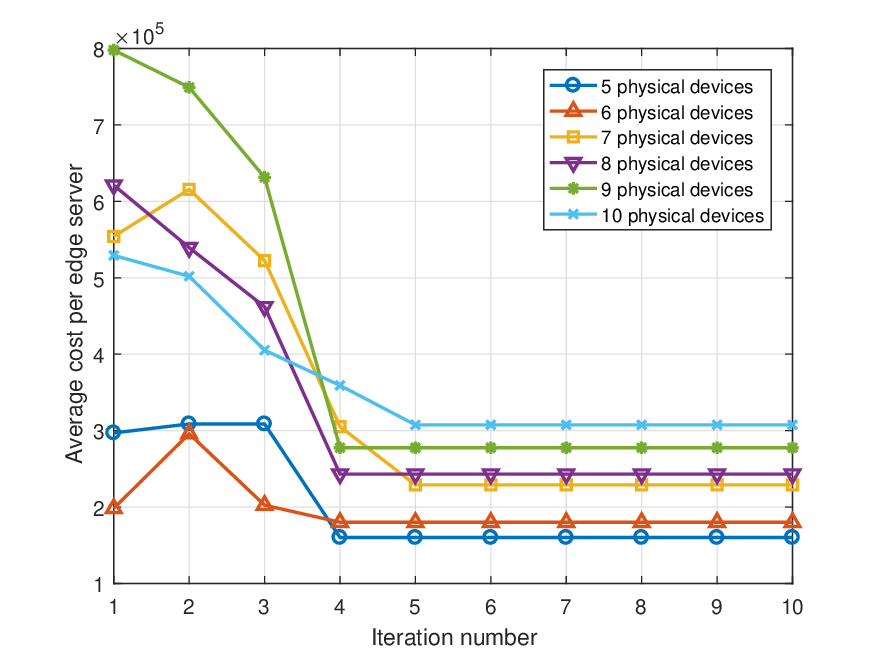}
    \caption{Convergence of average cost with a range of 1-5 components.}
    \label{fig4}
\end{figure}

\subsection{Average Cost Per edge server}

In this subsection, we analyze the average cost per edge server in two types of scenarios: physical device changing scenarios and edge server changing scenarios. Firstly, we evaluate the performance of physical device changing scenarios in Fig. \ref{fig5} and Fig. \ref{fig6}. Then, we evaluate the performance of different schemes under edge server changing scenarios in Fig. \ref{fig7} and Fig. \ref{fig8}.
Specifically, in Fig. \ref{fig5} and Fig. \ref{fig6}, the $x$-axis represents the number of physical devices, while $y$-axis depicts the average cost per edge server. In Fig. \ref{fig5}, we examine the performance where the number of components per physical device ranges from 1 to 3, the number of physical devices ranges from 5 to 10, and the number of edge servers remains constant at 6. On the other hand, Fig. \ref{fig6} demonstrates the performance where the number of components per physical device ranges from 1 to 5. The other parameters of Fig. \ref{fig6} remain the same as those of Fig. \ref{fig5}. We compare the performance of our proposed algorithm with the three baseline methods mentioned earlier.


\begin{figure}
    \centering
    \includegraphics[width=3in]{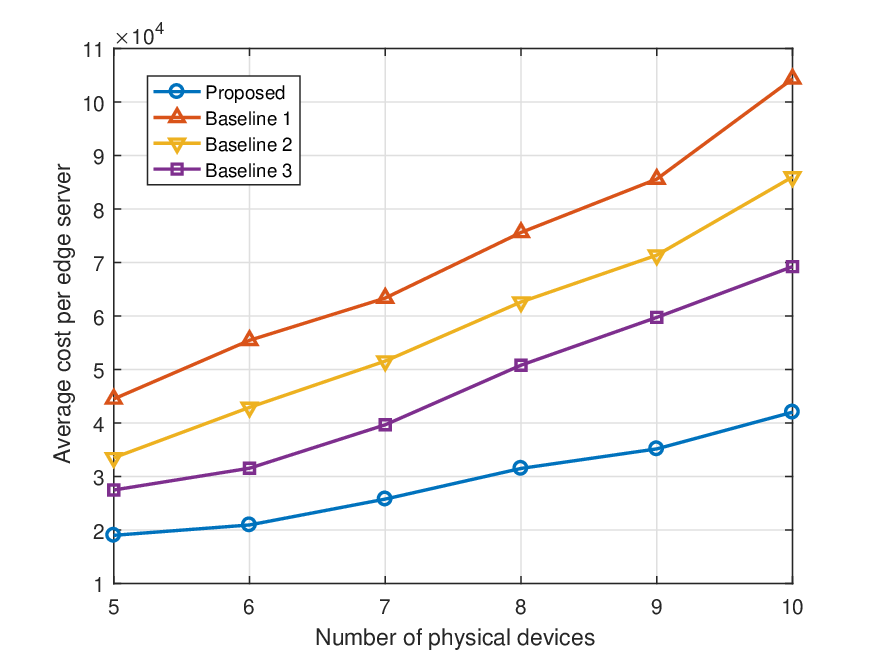}
    \caption{Average cost per edge server versus the number of physical devices, wherein the number of edge servers is 6 and the number of components per physical device ranges from 1 to 3.}
    \label{fig5}
\end{figure}


\begin{figure}
    \centering
    \includegraphics[width=3in]{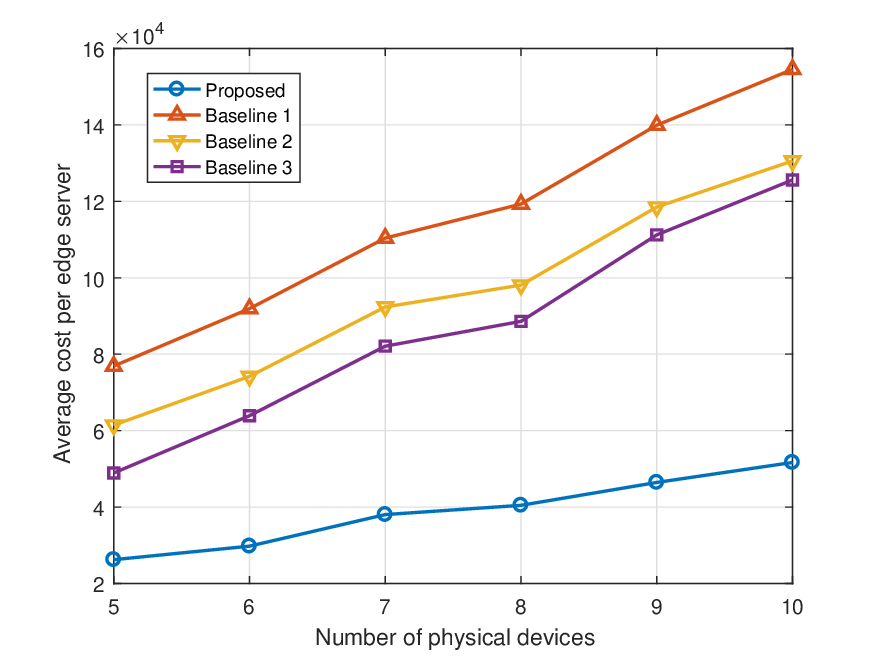}
    \caption{Average cost per edge server versus the number of physical devices, wherein the number of edge servers is 6 and the number of components per physical device ranges from 1 to 5.}
    \label{fig6}
\end{figure}

It can be observed from Fig. \ref{fig5} that the average cost per edge server increases as the number of physical devices increases for all compared algorithms. Among the baselines, Baseline 1, which can be considered as the starting state of Baseline 2, exhibits the poorest performance for each physical device number. Compared to Baseline 1, Baseline 2 achieves an average reduction of $19.55\%$ in average cost per edge server through its search procedure.
The reason behind this is that the conventional local search method can find a local optimal state, which is equal to or better than the starting state.
Furthermore, Baseline 3 provides an additional cost decrease of $20.37\%$. This indicates that the distance factor plays a significant role in the placement process. However, due to limitations in edge server capacities, it is not feasible to place all components into the nearest edge server. As a result, the feasible solutions obtained by this method are not optimal.
In contrast, the proposed algorithm achieves the lowest average cost compared to all the benchmark schemes. This is because the improved local search algorithm leverages regression to obtain insights into the structure of the search space, enabling the identification of a potentially better starting state for finding high-quality, near-optimal solutions. As a result, on average, the improved local search algorithm achieves an additional average cost decrease of $36.31\%$ compared to the Baseline 3. Moreover, a large number of physical devices results in a high average cost reduction.
For instance, in Fig. \ref{fig5}, compared to Baseline 3, the average cost reduction of the proposed method is $30.84\%$ in scenarios with 5 physical devices. This ratio expands to $39.30\%$ in scenarios with 10 physical devices.

In Fig. \ref{fig6}, we consider the scenario with a larger number of DT components than Fig. \ref{fig5}.
Comparing Fig. \ref{fig6} with Fig. \ref{fig5}, we can see that as the number of components per physical device increases, the average cost per edge server of all algorithms increases. Moreover, our developed scheme outperforms all the baselines significantly in terms of cost reduction. For instance, when the number of physical devices is 10, the proposed strategy saves cost by $66.56\%$, $60.43\%$, and $58.86\%$ on average, compared to Baseline 1, Baseline 2, and Baseline 3, respectively. Other similar trends to Fig. \ref{fig6} are omitted here for redundancy.

\begin{figure}
    \centering
    \includegraphics[width=3in]{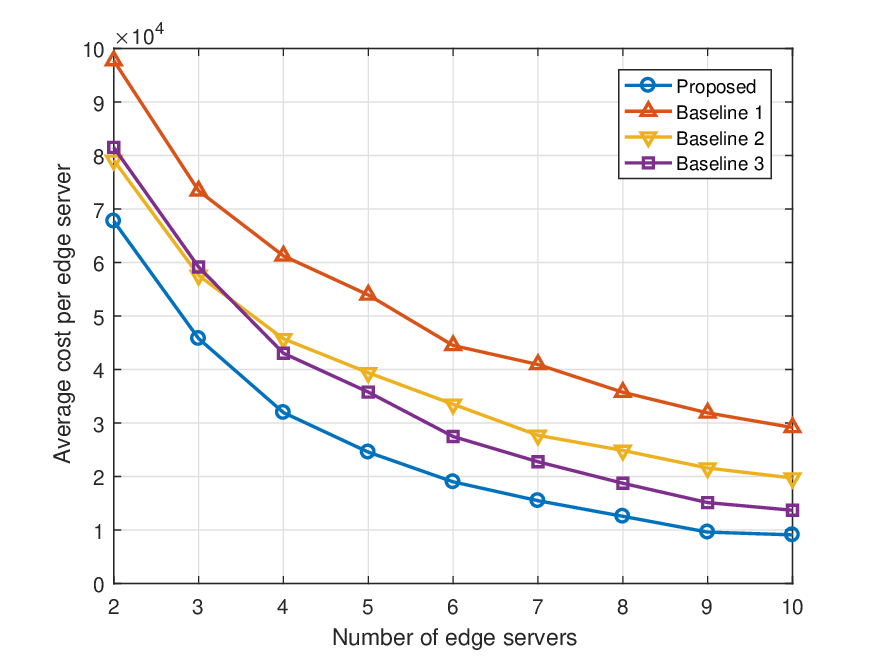}
    \caption{Average cost per edge server versus the number of edge servers, wherein the number of physical devices is 5.}
    \label{fig7}
\end{figure}

\begin{figure}
    \centering
    \includegraphics[width=3in]{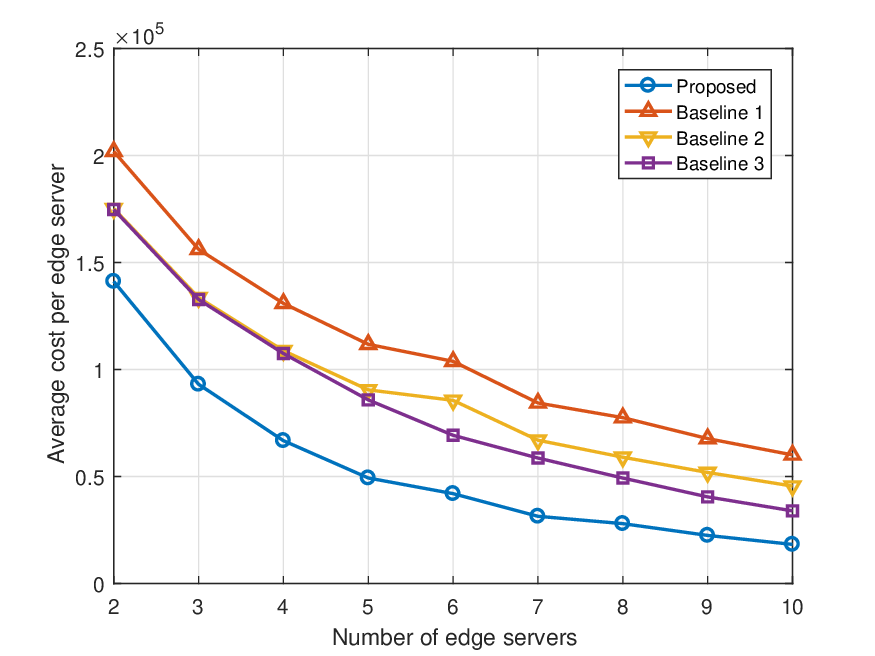}
    \caption{Average cost per edge server versus the number of edge servers, wherein the number of physical devices is 10.}
    \label{fig8}
\end{figure}

In Fig. \ref{fig7} and Fig. \ref{fig8}, we adjust the $x$-axis to represent the number of edge servers, ranging from 2 to 10, while keeping the number of physical devices constant in each figure.
Namely, the number of physical devices in Fig. \ref{fig7} is set to 5, while this parameter in Fig. \ref{fig8} is set to 10.
Additionally, the number of components per physical device in Fig. \ref{fig7} and Fig. \ref{fig8} ranges from 1 to 3.
It can be observed from both figures that the average cost per edge server decreases as the number of edge servers increases, but the rate of decrease slows down as the number of edge servers grows larger. This can be attributed to the fact that, with the number of physical devices remaining constant, the resource requirement also remains unchanged. Therefore, the trend of average cost per edge server appears as an approximate reciprocal function as the number of edge servers increases.
Moreover, among all the benchmark strategies, Baseline 1 results in the highest average cost per edge server. For example, when the number of edge servers is 2, the performance of Baseline 2 and Baseline 3 are comparable in scenarios with 5 and 10 physical devices, resulting in an overall cost decrease of $19.12\%$ and $13.35\%$, respectively, compared to Baseline 1.
The reason behind this is that when the number of edge servers is small, some edge servers face difficulties in placing all the components closest to them. As a result, these components need to be placed on other edge servers to satisfy the constraints, leading to intermediate performance.
However, as the number of edge servers increases, Baseline 3 outperforms Baseline 2.
This is because as the resources become sufficient, Baseline 3 has a higher chance of placing more components into the nearest edge servers. On the other hand, the Hillclimbing approach employed in Baseline 2 still suffers from the non-convex nature of the problem, resulting in poorer performance. Still, due to the information obtained from the regression, our proposed improved local search algorithm outperforms the early-mentioned baselines.

\subsection{Average Number of Searched States}
The average number of searched states refers to the total number of states explored by an algorithm in the search space $\cA$. For our proposed improved local search algorithm, this value is the accumulation of all data sizes $[Q_{t}]_{t = 1,2,\ldots,T}$. It is a vital criterion as it provides an indication of the scale explored by the algorithms in search spaces. In this subsection, we present the average number of searched states of our proposed algorithm and Baseline 2 in extensive scenario settings to demonstrate their effectiveness. It is important to note that Baseline 1 and Baseline 3 do not involve search procedures. Their starting states are already final solutions, which means they do not perform any exploration in the search space. Consequently, the average number of searched states for these baselines remains fixed at $T$. On this basis, we didn't present the average number of searched states for these baselines.

\begin{figure}
    \centering
    \includegraphics[width=3in]{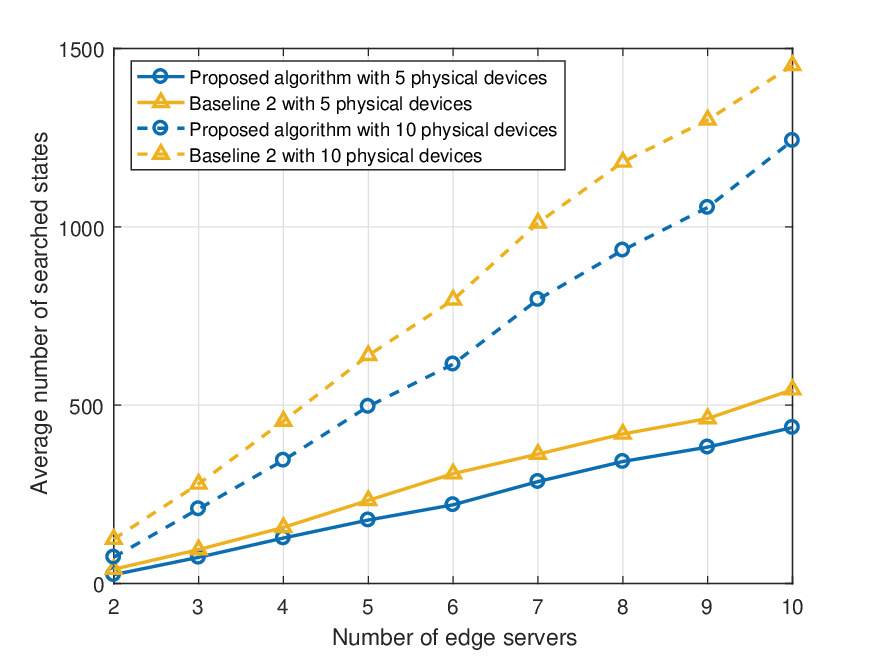}
    \caption{Average number of searched states versus the number of edge servers.}
    \label{fig9}
\end{figure}

Fig. \ref{fig9} shows the average number of searched states per scheme versus the number of edge servers. Therein, for each scheme, we consider two cases with different numbers of physical devices.
More specifically, solid lines are used to represent scenarios where the number of physical devices is 5, while dashed lines depict scenarios where the number of physical devices is 10.
It can be observed from Fig. \ref{fig9} that the average number of searched states increases almost linearly with the increasing of edge server numbers in all types of scenarios characterized by the number of physical devices. For scenarios with 10 physical devices, the increasing ratio is higher than that of scenarios with 5 physical devices. Furthermore, when combined with Fig. \ref{fig7} and Fig. \ref{fig8}, it can be seen that although the improved local search algorithm explored fewer states, its performance is better than that of Baseline 2. This is attributed to the improved local search algorithm converging after a few iterations, with its starting states being close to high-quality local-optimal states. As a result, the algorithm needs to search fewer states to find good solutions. Conversely, Baseline 2 experiences a comparable number of searched states in every iteration as it lacks insights into the structure of the search space.

\begin{figure}
    \centering
    \includegraphics[width=3in]{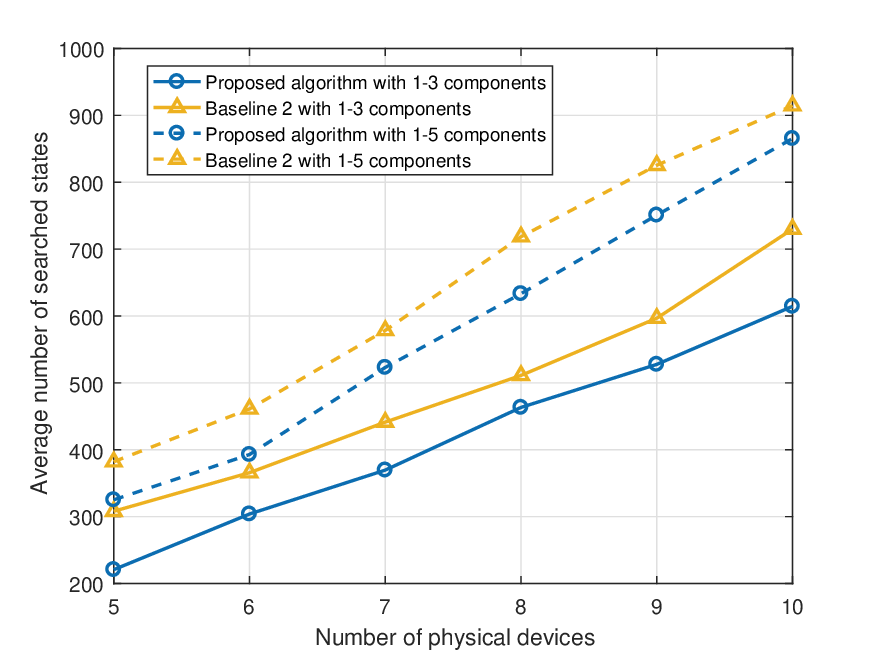}
    \caption{Average number of searched states versus the number of physical devices.}
    \label{fig10}
\end{figure}

Fig. \ref{fig10} investigates the number of physical devices on the amount of searched states per strategy, in which the number of edge servers is set to 6.
Similar to Fig. \ref{fig9}, we use solid lines to depict the number of searched states of the algorithms where the number of components per physical device ranges from 1 to 3. Correspondingly, dashed lines represent the performance of different schemes where the number of components ranges from 1 to 5.
It can be seen from Fig. \ref{fig10} that as the number of components increases, the required number of searched states also increases. Revisiting Fig. \ref{fig4} and Fig. \ref{fig5}, it can also be observed that the improved local search achieves better performance than the baseline even with less state space explored, demonstrating the effectiveness of our developed solution.

\section{Conclusions}
\label{sec6}
In this paper, we focused on the cost minimization-driven DT placement in WDTNs with the consideration of sustainability control for edge servers. To this end, we formulated the placement problem for WDTN as a chance-constrained integer programming problem. The formulated problem was challenging to solve due to its non-convexity and the feasibility checking issues. To make it tractable, we transformed the original minimization problem into an ILP and rigorously proved that the transformed problem was still NP-hard, making it difficult to obtain the optimal solution. To provide a time-efficient solution, we proposed an improved local search algorithm specifically designed to solve this ILP, which provided flexibility in balancing the time efficiency and performance guarantee. Extensive numerical results demonstrated the superiority of our developed solution compared to various baseline schemes in terms of cost savings as well as efficiency. Future work will involve studying the joint DT association and migration problems for WDTNs with accuracy consideration. Specifically, the impact of the interaction between the long-term WDTN maintenance and the performance enhancement of detailed digital twin tasks will be studied.

\bibliographystyle{IEEEtran}

\bibliography{reference}

\vspace{-1.3cm}
\begin{IEEEbiography}[{\includegraphics[width=1in,height=1.25in,clip,keepaspectratio]{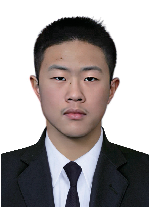}}]{Yuzhi Zhou}received his BEng (Hons.) in Communication Engineering from Liaocheng University (LCU), MSc in Information and Signal Processing from Nanjing University of Posts and Telecommunications (NJUPT) in 2019 and 2022, respectively. He is currently working toward the PhD degree at School of Science and Technology, Hong Kong Metropolitan University (HKMU). His research interests include wireless digital twin networks and their applications.
\end{IEEEbiography}

\vspace{-1.1cm}
\begin{IEEEbiography}[{\includegraphics[width=1in,height=1.25in,clip,keepaspectratio]{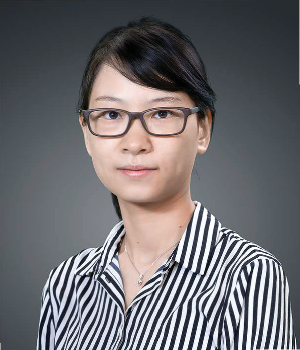}}]{Yaru Fu}(S'14-M'18)
received her Ph.D in Electronic Engineering from City University of Hong Kong (CityU) in 2018.  She is currently an Assistant Professor at the School of Science and Technology, Hong Kong Metropolitan University (HKMU). She is presently serving as an Associate Editor for the {\scshape IEEE Transactions on Cognitive Communications and Networking}, the {\scshape IEEE Internet of Things Journal}, the {\scshape IEEE Wireless Communications Letters}, the {\scshape IEEE Networking Letters}, and the {\scshape Springer Nature Computer Science}. She also serves as a Review Editor for the {\scshape Frontiers in Communications \& Networks}, a Guest Editor for the {\scshape Space: Science \& Technology}, and a Leading Guest Editor for the {\scshape Electronics} and the {\scshape IEEE Internet of Things Journal}.

Dr.\ Fu was honored with the 2021 Katie Shu Sui Pui Charitable Trust - Outstanding Research Publication Award (Gold Prize), 2022 Best Editor Award for {\scshape IEEE Wireless Communications Letters}, 2022 Katie Shu Sui Pui Charitable Trust - Excellent Research Publication Award, 2022 Exemplary Reviewer for the {\scshape IEEE Transactions on Communications} (fewer than 5\%), and 2023 President's Research Excellence Award. She was listed on the World's Top 2\% Scientists 2023 ranking by Stanford University in the United States. Her research interests include intelligent wireless communications and networking, mobile edge computing, and digital twins.
\end{IEEEbiography}

\vspace{-1.1cm}
\begin{IEEEbiography}[{\includegraphics[width=1in,height=1.25in,clip,keepaspectratio]{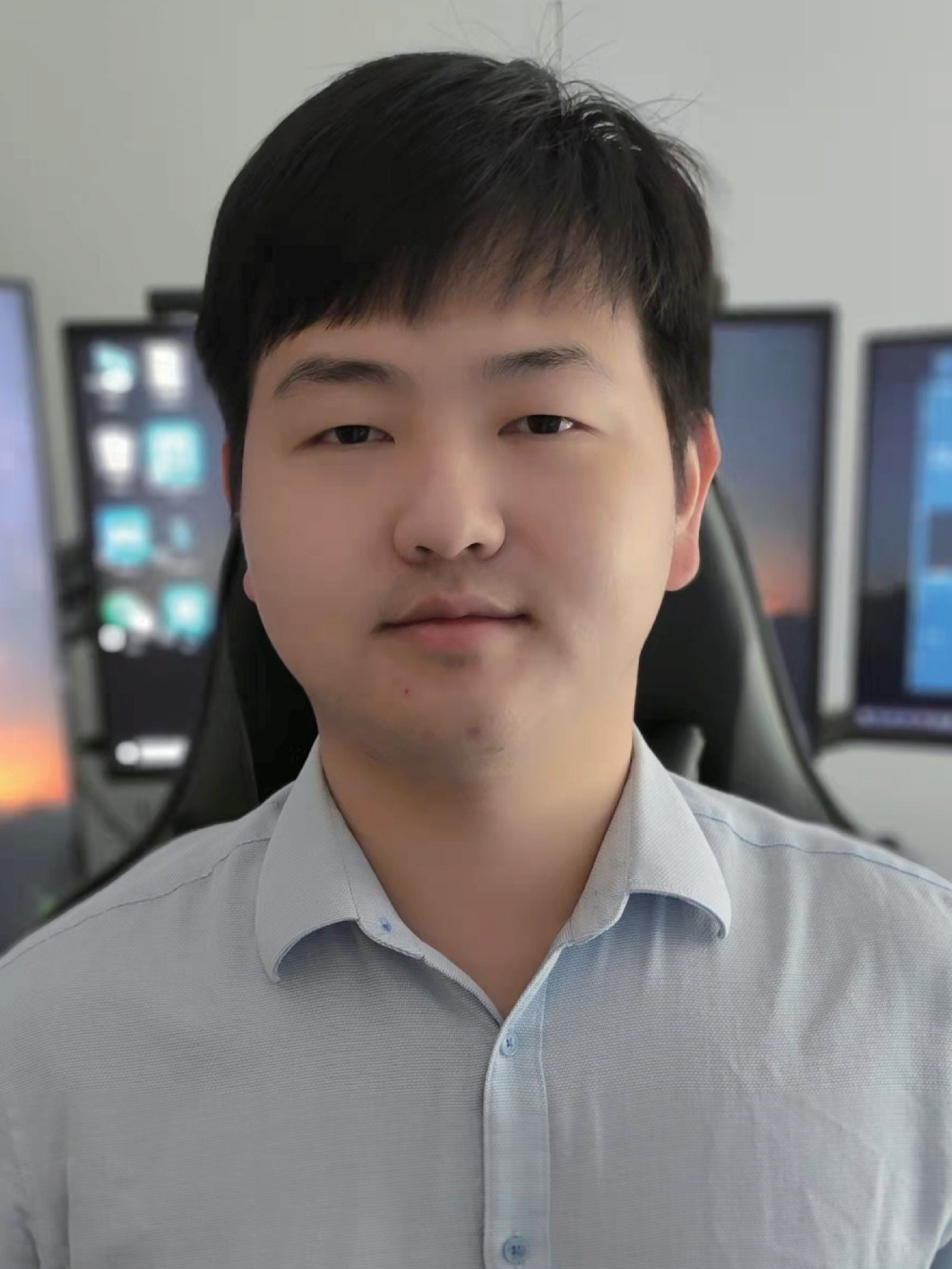}}]{Zheng Shi}obtained his Ph.D. degree in Electrical and Computer Engineering from University of Macau, Macao, in 2017. He is currently an Associate Professor with the School of Intelligent Systems Science and Engineering, Jinan University, Zhuhai, China. His current research interests include hybrid automatic repeat request, non-orthogonal multiple access, short-packet communications, intelligent reflecting surface, and Internet of Things.
\end{IEEEbiography}

\vspace{-1.1cm}
\begin{IEEEbiography}[{\includegraphics[width=1in,height=1.25in,clip,keepaspectratio]{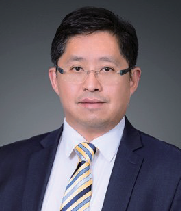}}]{Kevin Hung}is currently Associate Professor and the Head of Department of Electronic Engineering $\&$ Computer Science at the School of Science $\&$ Technology, Hong Kong Metropolitan University (HKMU). He is the principal investigator of several projects funded by the government and the university. Prior to joining HKMU, Dr. Hung was Assistant Project Manager at the Joint Research Centre for Biomedical Engineering at The Chinese University of Hong Kong (CUHK). He also worked as an Electronic Engineer at a medical device company. Dr. Hung's academic credentials include a B.Sc. from Queen’s University in Canada, and both an M.Phil. and Ph.D. from CUHK. His research interests include biosignal processing, quantum machine learning, biosystem simulation, and mobile health. Dr. Hung is currently serving as the Vice Chair of IEEE Hong Kong Section, Immediate Past Chair of the Electronics and Communications Section of IET Hong Kong, committee member of IET Hong Kong Branch, and Honorary Secretary of the Chinese Institute of Electronics (Hong Kong). He is also a founding officer of the IEEE Engineering in Medicine and Biology Society (EMBS) Hong Kong – Macau Joint Chapter, and served as its Chair in 2010. He is the founding Counsellor of the IEEE HKMU Student Branch.
\end{IEEEbiography}

\vspace{-1.1cm}
\begin{IEEEbiography}[{\includegraphics[width=1in,height=1.25in,clip,keepaspectratio]{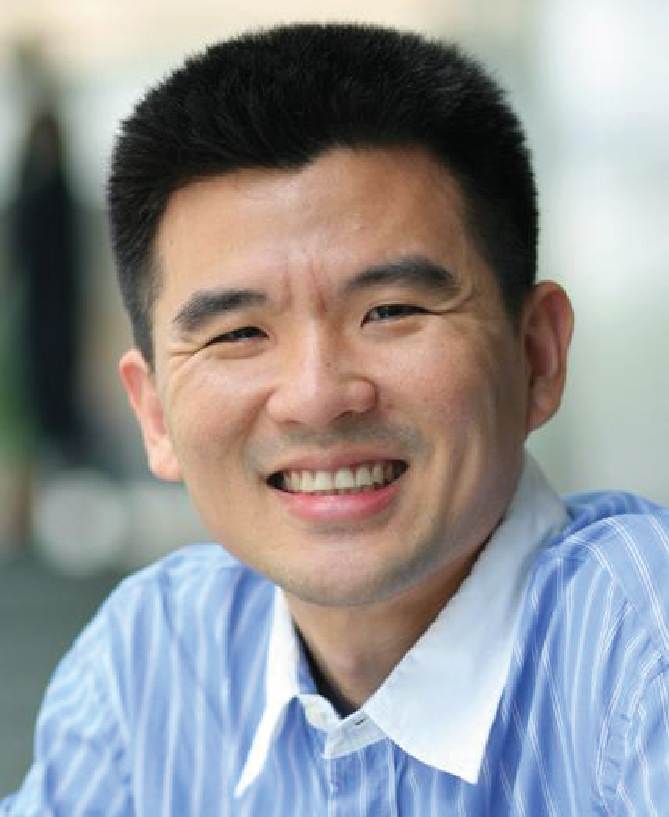}}]{Tony Q.S. Quek}(S'98-M'08-SM'12-F'18) received the B.E.\ and M.E.\ degrees in electrical and electronics engineering from the Tokyo Institute of Technology in 1998 and 2000, respectively, and the Ph.D.\ degree in electrical engineering and computer science from the Massachusetts Institute of Technology in 2008. Currently, he is the Cheng Tsang Man Chair Professor with Singapore University of Technology and Design (SUTD) and ST Engineering Distinguished Professor. He also serves as the Director of the Future Communications R\&D Programme, the Head of ISTD Pillar, and the AI on RAN Working Group Chair in AI-RAN Alliance. His current research topics include wireless communications and networking, network intelligence, non-terrestrial networks, open radio access network, and 6G.

Dr.\ Quek has been actively involved in organizing and chairing sessions, and has served as a member of the Technical Program Committee as well as symposium chairs in a number of international conferences. He is currently serving as an Area Editor for the {\scshape IEEE Transactions on Wireless Communications}.

Dr.\ Quek was honored with the 2008 Philip Yeo Prize for Outstanding Achievement in Research, the 2012 IEEE William R. Bennett Prize, the 2015 SUTD Outstanding Education Awards -- Excellence in Research, the 2016 IEEE Signal Processing Society Young Author Best Paper Award, the 2017 CTTC Early Achievement Award, the 2017 IEEE ComSoc AP Outstanding Paper Award, the 2020 IEEE Communications Society Young Author Best Paper Award, the 2020 IEEE Stephen O. Rice Prize, the 2020 Nokia Visiting Professor, the 2022 IEEE Signal Processing Society Best Paper Award, and the 2024 IIT Bombay International Award For Excellence in Research in Engineering and Technology. He is an IEEE Fellow, a WWRF Fellow, and a Fellow of the Academy of Engineering Singapore.
\end{IEEEbiography}

\vspace{-9cm}
\begin{IEEEbiography}[{\includegraphics[width=1in,height=1.25in,clip,keepaspectratio]{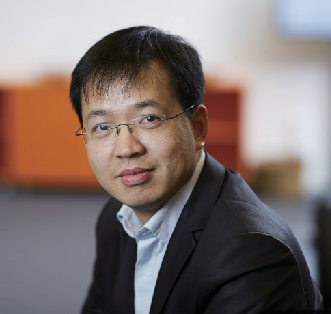}}]{Yan Zhang}(IEEE Fellow’20) is currently a Full Professor with the Department of Informatics, University of Oslo, Norway. He received the Ph.D. degree from the School of Electrical and Electronics Engineering, Nanyang Technological University, Singapore. His research interests include next-generation wireless networks leading to 6G, green and secure cyber-physical systems. Dr. Zhang is an Editor (or Area Editor, Senior Editor, Associate Editor) for several IEEE transactions/magazine. Since 2018, Prof. Zhang was a recipient of the global “Highly Cited Researcher” Award (Web of Science top $1\%$ most cited worldwide). He is Fellow of IEEE, Fellow of IET, elected member of Academia Europaea (MAE), elected member of the Royal Norwegian Society of Sciences and Letters (DKNVS), and elected member of Norwegian Academy of Technological Sciences (NTVA).
\end{IEEEbiography}

\end{document}